\documentclass[11pt]{article}

\usepackage[utf8]{inputenc}
\usepackage[T1]{fontenc}
\usepackage[english]{babel}
\usepackage[margin=1in]{geometry}
\usepackage{lmodern}
\usepackage{microtype}
\usepackage{authblk}
\usepackage{graphicx}
\usepackage{subcaption}
\usepackage{booktabs}
\usepackage{multirow}
\usepackage{amsmath,amssymb,amsfonts,amsthm}
\usepackage{mathrsfs}
\usepackage{braket}
\usepackage[round,authoryear]{natbib}
\usepackage{xcolor}
\usepackage{hyperref}

\hypersetup{
    colorlinks=true,
    linkcolor=blue!50!black,
    citecolor=blue!50!black,
    urlcolor=blue!50!black,
    pdftitle={A Novel Steganography Scheme Using Quantum Hilbert Transform},
    pdfauthor={Nitin Jha and Abhishek Parakh}
}

\theoremstyle{plain}
\newtheorem{theorem}{Theorem}
\newtheorem{proposition}[theorem]{Proposition}

\theoremstyle{definition}
\newtheorem{definition}{Definition}

\theoremstyle{remark}

\newcommand{\safeincludegraphics}[2][]{%
    \IfFileExists{#2}{\includegraphics[#1]{#2}}{%
        \fbox{\parbox[c][0.18\textheight][c]{0.9\linewidth}{%
            \centering Missing figure file:\\[0.5em]
            \texttt{\detokenize{#2}}%
        }}%
    }%
}

\title{A Novel Steganography Scheme Using Quantum Hilbert Transform}

\author[1]{Nitin Jha%
\thanks{Corresponding author:
\href{mailto:njha1@students.kennesaw.edu}
{njha1@students.kennesaw.edu}}}

\author[1]{Abhishek Parakh}

\affil[1]{Kennesaw State University, Marietta, GA, USA}

\date{}

\begin{document}

\maketitle

\begin{abstract}
The main goal of steganography is to transmit hidden messages in legitimate-looking communication messages. Phase-domain information hiding, however, has not been fully explored for quantum systems. This work introduces a finite-dimensional Quantum Hilbert Transform (QHT) as a unitary phase operator based on the Quantum Fourier Transform. Using this construction, we develop a QHT-based quantum steganography scheme that embeds classical bits as weak signed phase perturbations of quantum cover states. Bob recovers the hidden message through binary state discrimination, block aggregation, and classical error-correcting decoding.
\end{abstract}


\vspace{1em}

\section{Introduction}

Steganography is one of the oldest techniques of coordinating covert communication between two parties. The main goal of steganography is to hide a secret message in an otherwise legitimate-looking communication message carrier. In classical steganography, this carrier can be either text, image, video, audio, or communication waveform. A common design principle is to embed information in a domain where small perturbations are difficult to perceive or detect, while still being recoverable by an intended receiver. Steganography may be used in the national security context of military communication, espionage, or for intellectual property protection through digital rights management. Numerous techniques have been explored in the literature \citep{hameed2022literature}. In this paper, we propose the use of the quantum Hilbert transform for hiding information within quantum communication \citep{jha2025quantumhilbert}.

The Hilbert transform is one of the most significant signal and image processing tools used to generate an analytic signal from a real-valued signal by introducing a $\pm\pi/2$ phase shift to its frequency components. This transformation has been of immense importance in cases where one needs to extract the instantaneous phase, frequency, or amplitude of any given signal, crucial in applications ranging from envelope detection in biomedical signals to feature extraction in radar systems \citep{Hilbert2006,Hahn1996HilbertTI,oppenheim2021discrete,Feldman2008TheoreticalAA,khare2022vhers}. Hilbert transform enables the use of Single Sideband modulation (SSB), which reduces the bandwidth requirements and thus has even been exploited to embed imperceptible covert channels in audio streams \citep{Li2011,hilbertHiding2006}. The use of the Hilbert transform in steganography is often done by using the Hilbert-Huang transform. There have been several studies that use the classical Discrete Hilbert transform (DHT) and its adaptive extension, the Hilbert--Huang transform (HHT), to hide and detect secret information in a variety of signal and image domains \citep{hilberthuang,kandregula2009,wu2011detection,sharma2022hilbert}.

To exploit the usefulness of the Hilbert transform in the digital domain, several discrete Hilbert transform (DHT) algorithms have been developed. \citet{Kak1970DHT} lifting-based DHT and subsequent trigonometric variants allow efficient computation via the discrete Fourier transform (DFT) \citep{kak1977discrete,DHT2}. Hilbert-Huang transform is an adaptive, two-stage signal analysis technique that first decomposes any nonlinear, nonstationary time series into a finite set of data-driven Intrinsic Mode Functions (IMFs) via Empirical Mode Decomposition. The Hilbert transform is then applied to determine instantaneous amplitude and frequency information. This gives us the \textit{Hilbert Spectrum}. \citet{tan2007steganalysis} extended this idea to \textit{Enhanced Bit-Plane Complexity Segmentation (BPCS) steganography} by applying empirical mode decomposition and the Hilbert transform to bit-plane complexity difference sequences. \citet{wu2011detection} employed the Hilbert--Huang transform to break an image's pixel-value histogram into intrinsic mode functions and then extracted instantaneous amplitude and frequency descriptors---features that, fed into a machine-learning classifier, achieve very high accuracy in detecting embedded data across diverse image formats.

Subsequently, in the field of quantum computing, the quantum Fourier transform (QFT) has emerged as a core primitive for phase estimation, period finding, and quantum signal processing \citep{shor1999polynomial,zhouquant2017}. Quantum signal processing is a rapidly growing area of research, and several quantum analogues of classical transforms, such as the wavelet and cosine transforms, have been proposed \citep{pang2019signal,yin2021quantum,motlagh2024generalized}. This motivates the question of whether the phase-shifting behavior of the classical Hilbert transform can be lifted into the quantum setting and used as a primitive for quantum steganography. However, such a construction is not obtained by directly copying the classical Hilbert-transform mask. Classical Hilbert-transform and analytic-signal operations may suppress or amplify selected frequency components, while a coherent quantum operation must preserve norm and remain unitary.

This study presents a new QHT formulation that maps a quantum state into the Fourier basis, applies a phase transformation that depends on frequency and sign of the Fourier frequencies, and then returns the state to the computational basis, and finally applies the inverse QFT. The suggested QHT maintains the DC (zero-frequency component) and Nyquist (highest discrete-frequency component) modes in contrast to the classical transform, ensuring that the entire operation stays unitary. Building on this construction, we propose a QHT-based steganography protocol. This protocol hides a classical binary message by applying small QHT-domain phase perturbations to legitimate quantum cover states. A hidden bit is represented by the sign of a small unitary perturbation: one sign corresponds to a positive QHT-domain displacement, and the other corresponds to a negative displacement. Bob extracts the hidden information by distinguishing between these two possible perturbations using a binary quantum state discrimination procedure. To improve reliability, the protocol includes block-level aggregation and a classical error-correcting layer. To improve covertness, an optional whitening sequence randomizes the embedded signs and suppresses persistent first-order statistical bias.

\noindent
The purpose of this work is, therefore, to introduce two things:

\begin{enumerate}
    \item We define a finite-dimensional unitary quantum analogue of the Hilbert transform.
    \item We show how this operator can act as a coherent phase-domain primitive for quantum steganography.
\end{enumerate}

The resulting protocol connects a classical phase-based information-hiding idea with quantum state manipulation, while making the embedding operation physically valid as a unitary transformation. The remainder of this paper is organized as follows. Section~\ref{HT} reviews the continuous and discrete Hilbert transforms and introduces the quantum Fourier transform. Section~\ref{QHT} defines the proposed finite-dimensional QHT as a unitary Hilbert-phase operator and relates it to a fractional QHT-domain generator. Section~\ref{QSP} presents the QHT-based quantum steganography protocol, including the cover-state model, hidden-message encoding, whitening, embedding, and decoding. Section~\ref{Sec:Results} includes the numerical verification of the working of the quantum-steganography protocol defined in this work. Section~\ref{Conclusions} concludes the paper and discusses the relevance of the proposed model to covert-channel analysis in quantum-augmented networks.

\subsection{Discrete Hilbert Transform}
\label{HT}

The Discrete Hilbert transform (DHT) of a real-valued sequence $x[n]$ can be defined in the frequency domain using the Discrete Fourier transform (DFT). The DHT is obtained by multiplying the DFT coefficients by a phase-only Hilbert multiplier. This is distinct from analytic-signal construction, which zeros the negative-frequency components and doubles the positive-frequency components. Let $x[n]$ be a length-$N$ sequence with DFT,

\begin{equation}
X[k]
=
\sum_{n=0}^{N-1}
x[n]e^{-2\pi \mathrm{i}kn/N}.
\end{equation}

The DHT is obtained by multiplying each Fourier coefficient by the discrete Hilbert-transform phase multiplier

\begin{equation}
m_{\mathcal{H}}[k] =
\begin{cases}
0, & k=0,\\
-\mathrm{i}, & 1\leq k\leq \frac{N}{2}-1,\\
0, & k=\frac{N}{2} \text{ for even } N,\\
+\mathrm{i}, & \frac{N}{2}+1\leq k\leq N-1.
\end{cases}
\label{eq:dht-phase-mask}
\end{equation}

Then,

\begin{equation}
\widehat{X}[k]
=
m_{\mathcal{H}}[k]X[k],
\end{equation}

and the discrete Hilbert transform is

\begin{equation}
\widehat{x}[n]
=
\frac{1}{N}
\sum_{k=0}^{N-1}
\widehat{X}[k]e^{2\pi \mathrm{i}kn/N}.
\end{equation}

The analytic signal is related to, but distinct from, the DHT. It is constructed as

\begin{equation}
x_{\rm an}[n]
=
x[n]+\mathrm{i}\widehat{x}[n].
\end{equation}

Equivalently, it can be formed using the analytic-signal mask

\begin{equation}
H_{\rm an}[k] =
\begin{cases}
1, & k=0,\\
2, & 1\leq k\leq \frac{N}{2}-1,\\
1, & k=\frac{N}{2} \text{ for even } N,\\
0, & \frac{N}{2}+1\leq k\leq N-1.
\end{cases}
\end{equation}

This analytic-signal mask is not unitary because it contains entries with magnitude $0$ and $2$. Therefore, it cannot be directly promoted to a closed-system quantum operation. The QHT introduced in Section~\ref{QHT} instead uses a phase-only multiplier with entries of unit magnitude.

\subsection{Quantum Fourier Transform}

Discrete Fourier transform (DFT) is defined to take a unit vector,
$\hat{x}=(x_0,x_1,\ldots,x_{N-1})\in\mathbb{C}^N$, and produce another unit vector,
$\hat{y}=(y_0,y_1,\ldots,y_{N-1})$. The output component $\hat{y}_k$ can be written as \citep{zhouquant2017}

\begin{equation}
y_k
=
\frac{1}{\sqrt{N}}
\sum_{j=0}^{N-1}
e^{2\pi \mathrm{i}jk/N}x_j,
\qquad
k=0,1,\ldots,N-1.
\label{dft}
\end{equation}

The Quantum Fourier Transform (QFT) operates similarly to the discrete Fourier transform but in amplitudes. Mathematically,

\begin{equation}
\sum_{j=0}^{N-1}x_j\ket{j}
\longrightarrow
\sum_{k=0}^{N-1}y_k\ket{k}.
\label{QFT}
\end{equation}

From Equation~\eqref{QFT}, we can see that the QFT takes a quantum state from the computational basis to the Fourier basis. The state information is encoded in the Fourier coefficients:

\begin{equation}
\operatorname{QFT}(\ket{\psi})
=
\frac{1}{\sqrt{N}}
\sum_{j=0}^{N-1}
\left(
\sum_{k=0}^{N-1}
x_k e^{2\pi \mathrm{i}jk/N}
\right)
\ket{j},
\label{qft-fourier}
\end{equation}

where

\begin{equation}
\beta_j
=
\sum_{k=0}^{N-1}
x_k e^{2\pi \mathrm{i}jk/N}
\end{equation}

are the Fourier coefficients. We will use the phase information extracted from these Fourier coefficients in the next section when we define the quantum Hilbert transform.

\section{Proposed Quantum Hilbert Transform}
\label{QHT}

One of the two key contributions of this work is to develop a quantum analogue of the classical Hilbert Transform. Classically, the Hilbert transform is most naturally understood in the Fourier domain, where it applies a frequency-sign-dependent phase shift. Positive frequency components are multiplied by $-\mathrm{i}$, while negative frequency components are multiplied by $+\mathrm{i}$, where $\mathrm{i}=\sqrt{-1}$. The zero-frequency component is treated separately. Therefore, the Hilbert transform does not primarily act by changing the magnitude of a signal, but by changing the relative phase of its frequency components.

In the classical setting, the zero-frequency component may be mapped to zero. However, such an operation is not unitary and therefore cannot directly represent the evolution of a closed quantum system. For this reason, we define the quantum Hilbert transform as a finite-dimensional unitary phase operator that preserves the exceptional Fourier modes while applying the Hilbert-transform phase convention to the positive and negative Fourier modes.

Let $N=2^n$ for an $n$-qubit system, with $n\geq2$. Let

\begin{equation}
\left\{
\ket{0},
\ket{1},
\ldots,
\ket{N-1}
\right\}
\end{equation}

denote the computational basis. Thus, an arbitrary quantum state can be written as

\begin{equation}
\ket{\psi}
=
\sum_{j=0}^{N-1}
x_j\ket{j},
\qquad
\sum_{j=0}^{N-1}|x_j|^2=1.
\label{psi}
\end{equation}

Let $F_N$ denote the $N$-dimensional quantum Fourier transform, defined by

\begin{equation}
F_N\ket{j}
=
\frac{1}{\sqrt{N}}
\sum_{k=0}^{N-1}
e^{2\pi \mathrm{i}jk/N}\ket{k}.
\end{equation}

Applying the QFT to $\ket{\psi}$ from Equation~\eqref{psi} gives

\begin{equation}
F_N\ket{\psi}
=
\sum_{k=0}^{N-1}
X_k\ket{k},
\end{equation}

where

\begin{equation}
X_k
=
\frac{1}{\sqrt{N}}
\sum_{j=0}^{N-1}
x_j e^{2\pi \mathrm{i}jk/N}
\end{equation}

are the Fourier-basis amplitudes of the state.

We can now define a diagonal Hilbert phase operator $D_{\mathcal H}$ in the Fourier basis. For even $N$, let

\begin{equation}
D_{\mathcal H}\ket{k}
=
h_k\ket{k},
\end{equation}

where

\begin{equation}
h_k =
\begin{cases}
1, & k=0,\\
-\mathrm{i}, & 1\leq k\leq \frac{N}{2}-1,\\
1, & k=\frac{N}{2},\\
+\mathrm{i}, & \frac{N}{2}+1\leq k\leq N-1.
\end{cases}
\end{equation}

Here, $k=0$ represents the DC mode,
$1\leq k\leq N/2-1$ represents the positive Fourier modes,
$k=N/2$ is the Nyquist mode for even $N$, and
$N/2+1\leq k\leq N-1$ represents the negative Fourier modes.
The positive modes therefore receive a $-\pi/2$ phase shift, while the negative modes receive a $+\pi/2$ phase shift. The DC and Nyquist modes are preserved so that the operation remains unitary.

\begin{definition}[Quantum Hilbert Transform]
The quantum Hilbert transform is defined as

\begin{equation}
U_{\mathcal H}
=
F_N^\dagger D_{\mathcal H}F_N.
\end{equation}

Thus, the QHT maps the input state to the Fourier basis, applies the Hilbert-transform phase rule to the Fourier modes, and finally maps the state back to the computational basis.
\end{definition}

The action of the QHT on an input state can be written as

\begin{equation}
U_{\mathcal H}\ket{\psi}
=
F_N^\dagger D_{\mathcal H}F_N\ket{\psi}.
\end{equation}

Equivalently,

\begin{equation}
F_N\ket{\psi}
=
\sum_{k=0}^{N-1}
X_k\ket{k},
\end{equation}

\begin{equation}
D_{\mathcal H}F_N\ket{\psi}
=
\sum_{k=0}^{N-1}
h_kX_k\ket{k},
\end{equation}

and therefore,

\begin{equation}
U_{\mathcal H}\ket{\psi}
=
F_N^\dagger
\left(
\sum_{k=0}^{N-1}
h_kX_k\ket{k}
\right).
\end{equation}

The above concludes the formation of the quantum Hilbert transform (QHT), which is a unitary finite-dimensional analogue of the classical Hilbert transform. In the classical Hilbert transform, the zero-frequency component is annihilated. In the quantum setting, mapping a basis component to zero would violate the unitary property. Therefore, the QHT preserves the DC component. Similarly, for even $N$, the Nyquist mode is also preserved.

\subsection{Fractional QHT Generator}
\label{sec:fractional-qht}

The QHT can also be described through a skew-Hermitian generator. For even $N$, define the diagonal operator $A_{\mathcal H}$ in the Fourier basis by

\begin{equation}
A_{\mathcal H}\ket{k}
=
a_k\ket{k},
\end{equation}

where

\begin{equation}
a_k =
\begin{cases}
0, & k=0,\\
-\mathrm{i}, & 1\leq k\leq N/2-1,\\
0, & k=N/2,\\
+\mathrm{i}, & N/2+1\leq k\leq N-1.
\end{cases}
\label{eq:qht-generator-mask-main}
\end{equation}

Since all nonzero entries of $A_{\mathcal H}$ are purely imaginary and occur as Hilbert-transform phase directions, we have

\begin{equation}
A_{\mathcal H}^{\dagger}
=
-A_{\mathcal H}.
\end{equation}

Therefore,

\begin{equation}
G_{\mathcal H}
=
F_N^\dagger A_{\mathcal H}F_N
\end{equation}

is also skew-Hermitian. The Hilbert-phase operator $D_{\mathcal H}$ satisfies

\begin{equation}
D_{\mathcal H}
=
\exp\left(
\frac{\pi}{2}A_{\mathcal H}
\right).
\end{equation}

This follows because the entries $a_k=-\mathrm{i}$ on the positive Fourier modes generate the phase factor $e^{-\mathrm{i}\pi/2}=-\mathrm{i}$, while the entries $a_k=+\mathrm{i}$ on the negative Fourier modes generate the phase factor $e^{+\mathrm{i}\pi/2}=+\mathrm{i}$. The entries $a_k=0$ on the DC and Nyquist modes generate the identity phase factor. Thus,

\begin{equation}
U_{\mathcal H}
=
F_N^\dagger D_{\mathcal H}F_N
=
\exp\left(
\frac{\pi}{2}G_{\mathcal H}
\right).
\end{equation}

This representation shows that the QHT is generated by $G_{\mathcal H}$ and that smaller signed powers of this generator can be interpreted as fractional QHT-domain displacements. We define the Hermitian observable

\begin{equation}
K_{\mathcal H}
=
\mathrm{i}G_{\mathcal H}.
\end{equation}

Since $G_{\mathcal H}$ is skew-Hermitian, $K_{\mathcal H}$ is Hermitian. The steganographic embedding unitary can then be written as

\begin{equation}
V_{\alpha}(\epsilon)
=
\exp\left(
\alpha\epsilon G_{\mathcal H}
\right)
=
\exp\left(
-\mathrm{i}\alpha\epsilon K_{\mathcal H}
\right).
\end{equation}

The proposed QHT is not identical to the classical Hilbert transform on the full Fourier space. In particular, the DC and Nyquist modes are preserved rather than annihilated. Consequently, $U_{\mathcal H}^{2}$ acts as $-I$ only on the nonzero non-Nyquist Fourier subspace, while the exceptional DC and Nyquist modes are left unchanged. This distinction is necessary to maintain the unitary property and is the main difference between the classical Hilbert transform and the finite-dimensional quantum Hilbert-phase operator developed here.

\section{Quantum Steganography Protocol}
\label{QSP}

In this section, we discuss the formulation of a steganography protocol based on the quantum Hilbert transform (QHT) in the previous section. The main idea is to hide information in weak QHT-domain phase perturbations of legitimate quantum states. Unlike classical steganography, where secret information may be hidden in pixels, audio samples, or frequency coefficients, the proposed protocol hides information in the phase structure of quantum states.

\subsection{Cover-State Model}
\label{Sec:coverstate}

Let Alice and Bob communicate using a sequence of legitimate quantum cover states

\begin{equation}
\ket{c_t},
\qquad
t=1,2,\ldots,T,
\end{equation}

where each $\ket{c_t}$ belongs to an $N$-dimensional Hilbert space with $N=2^n$. This is monitored by a passive warden, Wendy. Wendy's goal is to analyze the transmitted states to detect the presence of a steganographic payload without disrupting the legitimate network traffic. Each cover state can be written as

\begin{equation}
\ket{c_t}
=
\sum_{j=0}^{N-1}
c_{t,j}\ket{j},
\qquad
\sum_{j=0}^{N-1}|c_{t,j}|^2=1.
\end{equation}

For steganographic embedding, we use the QHT generator $G_{\mathcal H}$ defined in Section~\ref{sec:fractional-qht}. Equivalently, using the Hermitian observable

\begin{equation}
K_{\mathcal H}
=
\mathrm{i}G_{\mathcal H},
\end{equation}

the embedding operation can be interpreted as a signed phase displacement generated by $K_{\mathcal H}$. Since $G_{\mathcal H}$ is skew-Hermitian, $\exp(\theta G_{\mathcal H})$ is unitary for any real $\theta$.

\subsection{Hidden Message Encoding}
\label{Sec:hiddenmess}

Let the hidden message be a binary string,

\begin{equation}
M=m_1m_2\ldots m_\ell,
\qquad
m_i\in\{0,1\}.
\end{equation}

Since practical quantum hardware is noisy, Alice first encodes the hidden message using a classical binary error-correcting code to improve reliability. Let

\begin{equation}
\operatorname{Enc}:
\{0,1\}^{\ell}
\longrightarrow
\{0,1\}^{\ell'}
\end{equation}

denote the encoder. The encoded hidden sequence is

\begin{equation}
\mathbf{b}
=
\operatorname{Enc}(M)
=
(b_1,b_2,\ldots,b_{\ell'}),
\qquad
b_t\in\{0,1\}.
\end{equation}

Here, $\ell'$ denotes the length of the encoded hidden sequence after error-correcting encoding. This coding layer is classical because the hidden payload is a classical bit string.

In the basic one-carrier-per-encoded-bit version of the protocol, the $t$-th encoded bit $b_t$ is embedded into the corresponding cover state $\ket{c_t}$, so that $T=\ell'$. If block spreading is used, each encoded bit may instead be embedded across multiple carrier states, as described in the decoding section. In that case, for a fixed block length $L$, the total number of transmitted cover states is $T=L\ell'$.

If a shared stego-key is available, Alice and Bob may use it to generate a pseudo-random whitening sequence

\begin{equation}
r_1,r_2,\ldots,r_{\ell'},
\qquad
r_t\in\{0,1\}.
\end{equation}

The encoded bit is then randomized as

\begin{equation}
\widetilde{b}_t
=
b_t\oplus r_t.
\end{equation}

For a keyless version of the protocol, one simply sets

\begin{equation}
\widetilde{b}_t=b_t.
\end{equation}

The randomized bit is mapped to a sign value,

\begin{equation}
\alpha_t
=
(-1)^{\widetilde{b}_t}.
\end{equation}

Thus,

\begin{equation}
\alpha_t =
\begin{cases}
+1, & \widetilde{b}_t=0,\\
-1, & \widetilde{b}_t=1.
\end{cases}
\end{equation}

We now use the QHT to perform embedding. Let $\epsilon>0$ be a small embedding strength. For each cover state $\ket{c_t}$, Alice applies the QHT-domain steganographic unitary

\begin{equation}
V_{\alpha_t}(\epsilon)
=
\exp\left(
\alpha_t\epsilon G_{\mathcal H}
\right).
\end{equation}

The transmitted stego state is

\begin{equation}
\ket{s_t}
=
V_{\alpha_t}(\epsilon)\ket{c_t}.
\label{eq:stego-state}
\end{equation}

In the Fourier basis, let

\begin{equation}
F_N\ket{c_t}
=
\sum_{k=0}^{N-1}
C_{t,k}\ket{k}.
\end{equation}

The embedding operation applies the transformation

\begin{equation}
C_{t,k}
\longrightarrow
e^{\alpha_t\epsilon a_k}C_{t,k}.
\end{equation}

Using Equation~\eqref{eq:qht-generator-mask-main}, the DC and Nyquist modes remain unchanged, while the positive and negative Fourier modes acquire opposite-signed phase shifts:

\begin{equation}
C_{t,k}
\longrightarrow
e^{-\mathrm{i}\alpha_t\epsilon}C_{t,k},
\qquad
1\leq k\leq\frac{N}{2}-1,
\end{equation}

and

\begin{equation}
C_{t,k}
\longrightarrow
e^{+\mathrm{i}\alpha_t\epsilon}C_{t,k},
\qquad
\frac{N}{2}+1\leq k\leq N-1.
\end{equation}

Thus, the hidden bit is encoded as a weak signed QHT-domain phase displacement. If $\widetilde{b}_t=0$, then $\alpha_t=+1$; if $\widetilde{b}_t=1$, then $\alpha_t=-1$.

\subsection{Decoding and Extraction}
\label{Sec:decoding}

Bob receives the sequence of stego states

\begin{equation}
\ket{s_1},
\ket{s_2},
\ldots,
\ket{s_T}.
\end{equation}

We assume that Bob has access to the expected cover-state model,

\begin{equation}
\left\{
\ket{c_t}
\right\}_{t=1}^{T},
\end{equation}

as well as the embedding strength $\epsilon$ and the QHT-domain embedding operator. For each received state, Bob must determine whether Alice applied the positive or negative embedding operation. For the $t$-th cover state, we define the two possible stego states as

\begin{equation}
\ket{\psi_{+}^{(t)}}
=
V_{+}(\epsilon)\ket{c_t},
\qquad
\ket{\psi_{-}^{(t)}}
=
V_{-}(\epsilon)\ket{c_t}.
\end{equation}

The extraction problem for each carrier state can therefore be written as a binary quantum state-discrimination problem:

\begin{equation}
H_0:
\ket{s_t}
=
\ket{\psi_{+}^{(t)}},
\end{equation}

and

\begin{equation}
H_1:
\ket{s_t}
=
\ket{\psi_{-}^{(t)}}.
\end{equation}

The hypothesis $H_0$ corresponds to a positive QHT-domain perturbation, while $H_1$ corresponds to a negative QHT-domain perturbation. Since the hidden bit is embedded in the sign of this perturbation, deciding between $H_0$ and $H_1$ allows Bob to recover the hidden bit. Equivalently, Bob estimates the sign of the weak phase displacement introduced in the QHT domain. For small $\epsilon$, the received state differs from the cover state mainly through the first-order perturbation generated by $G_{\mathcal H}$. Thus, the decoding problem can be interpreted as deciding whether the perturbation was applied in the $+G_{\mathcal H}$ direction or the $-G_{\mathcal H}$ direction.

The ideal minimum-error version of this binary decision problem is given by the Helstrom measurement. In the main protocol, the binary test is treated abstractly as a decoding step that produces a noisy sign estimate for each carrier state. A single carrier state may not provide enough information for reliable retrieval when $\epsilon$ is small. Therefore, we can spread the hidden information over a block of $L$ carrier states.

Let $B_i$ denote the block of carrier states associated with the $i$-th encoded hidden bit. For each $t\in B_i$, Bob obtains a noisy sign estimate,

\begin{equation}
Y_t\in\{+1,-1\}.
\end{equation}

Here, $Y_t=+1$ means Bob prefers $H_0$, and $Y_t=-1$ means Bob prefers $H_1$. Bob then aggregates the sign estimates over the block:

\begin{equation}
S_i
=
\sum_{t\in B_i}Y_t.
\end{equation}

The encoded bit is estimated by majority decision:

\begin{equation}
\widehat{\widetilde{b}}_i =
\begin{cases}
0, & S_i\geq0,\\
1, & S_i<0.
\end{cases}
\end{equation}

Thus, if the aggregate sign is nonnegative, Bob decides that the positive embedding sign was used and decodes the embedded bit as $0$. If the aggregate sign is negative, Bob decides that the negative embedding sign was used and decodes the embedded bit as $1$.

If a whitening sequence was used, Bob removes the whitening after all block-level decisions have been finalized. If $r_i$ is the whitening bit associated with the $i$-th encoded position, Bob computes

\begin{equation}
\widehat{b}_i
=
\widehat{\widetilde{b}}_i\oplus r_i.
\end{equation}

The encoded bit sequence is then passed through the error-correcting decoder:

\begin{equation}
\widehat{M}
=
\operatorname{Dec}
\left(
\widehat{b}_1,
\widehat{b}_2,
\ldots,
\widehat{b}_{\ell'}
\right),
\end{equation}

where $\ell'$ denotes the length of the encoded hidden sequence after error-correcting encoding.

The proposed decoding model assumes that the receiver has access to the expected cover-state description for each transmission interval. Therefore, the protocol is not intended for arbitrary unknown quantum states. Instead, it applies to structured quantum-augmented network traffic in which routine protocol states are generated from a known schedule, public metadata, or a pre-established cover-state codebook. For example, in a quantum-augmented microgrid, periodic synchronization, heartbeat, controller identification, or status-confirmation frames may be mapped to predetermined quantum cover states. In such a scenario, where nodes can reconstruct quantum payloads as part of routine traffic, the proposed protocol can be used to study how sensitive information might be hidden through weak QHT-domain perturbations whose detectability depends on the embedding strength, whitening, channel noise, and the number of carrier states observed by the warden.

\subsection{Working Example}
\label{Sec:example}

To illustrate an application of the protocol, we use the example of a networked microgrid in which two nodes are secretly compromised and attempt to leak sensitive control information through periodic control commands.

Assume that two local controllers are compromised: for example, the controller in Microgrid~$1$ and the controller in Microgrid~$3$. The compromised controller in Microgrid~$1$ embeds a low-rate hidden payload into its otherwise legitimate routine reports, while the compromised controller in Microgrid~$3$ acts as the covert receiver and decodes the hidden information. The operator or network monitor continues to observe traffic that appears to be ordinary controller communication because the visible communication pattern and message timing remain unchanged.

Thus, the cover states

\begin{equation}
S_t,S_{t+1},S_{t+2},\ldots
\end{equation}

represent the apparent routine packets, while the hidden payload is carried through weakly structured perturbations embedded within them. Based on the QHT-steganography protocol, the hidden information is encoded as a subtle QHT-domain perturbation of the transmitted cover state.

Let $\ket{c_t}$ denote the legitimate quantum state corresponding to a routine controller packet at time $t$. The compromised node prepares the transmitted stego state as

\begin{equation}
\ket{s_t}
=
\exp\left(
\alpha_t\epsilon G_{\mathcal H}
\right)
\ket{c_t},
\end{equation}

where $G_{\mathcal H}$ is the QHT generator, $\epsilon$ is a small embedding strength, and $\alpha_t\in\{+1,-1\}$ encodes the hidden bit. Since $\epsilon$ is small, the stego state remains close to the original cover state and therefore preserves the appearance of normal communication. This shows how two compromised parties can use regular quantum communication exchange to exploit and hide secret messages that might otherwise be blocked by the network administrator. Figure~\ref{fig:stegomicrogrid} shows a schematic representation of one possible example of the proposed steganography protocol.

\begin{figure}[!htbp]
    \centering
    \safeincludegraphics[width=\linewidth]
    {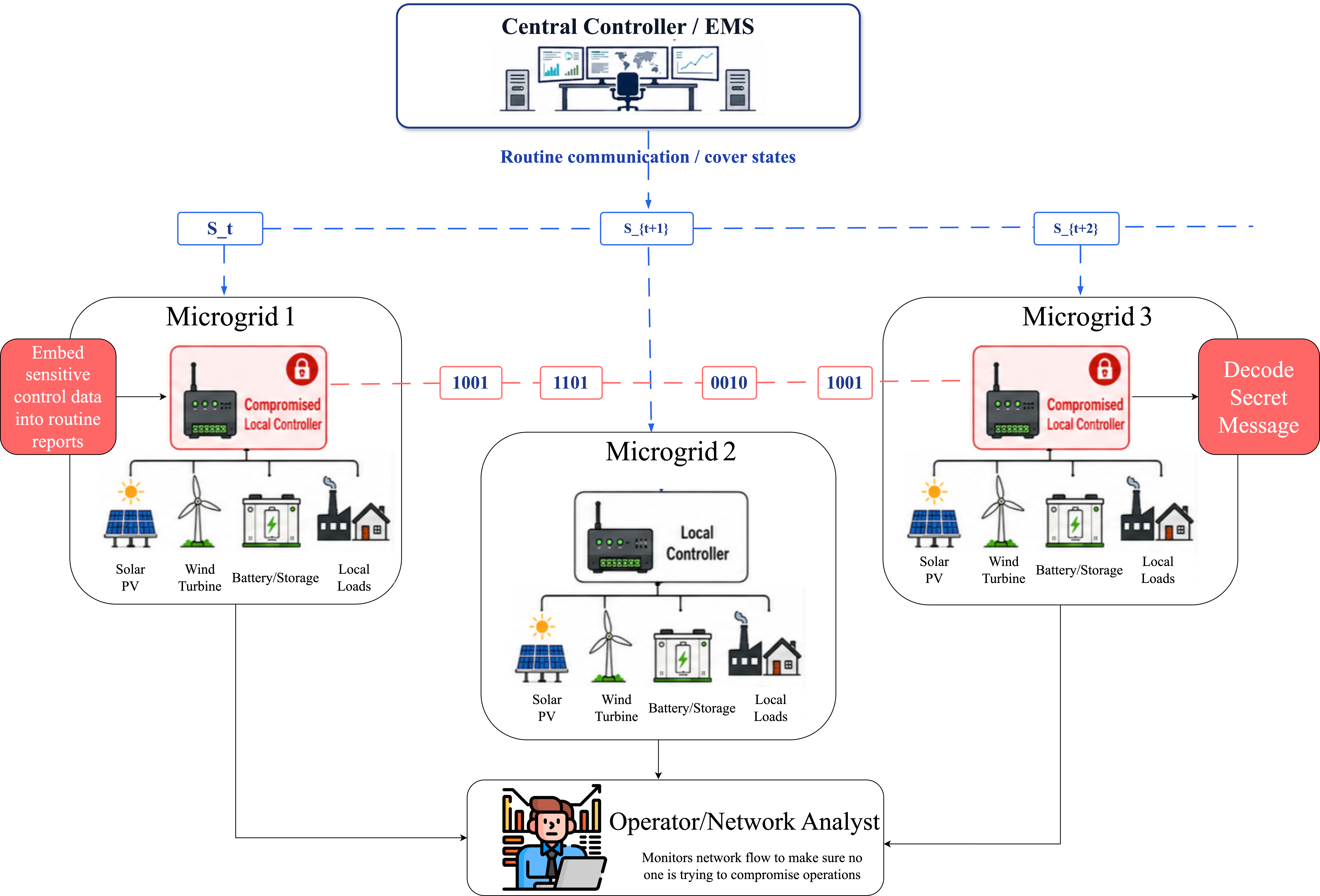}
    \caption{A schematic diagram showing how the QHT-based steganography protocol can exploit quantum communication to embed hidden information across two compromised parties in a microgrid. A synchronized attack between two parties can be conducted while the network analyst remains unable to infer the presence of hidden information among the routine control commands transmitted as quantum states.}
    \label{fig:stegomicrogrid}
\end{figure}

\section{Results}
\label{Sec:Results}

In this section, we present experimental results to cross-validate the operation of the protocol. The goal is to study the tradeoff among Bob's decoding reliability, the disturbance introduced into the cover state, and the detectability of the embedding by a warden. We evaluate the effects of the embedding strength $\epsilon$, spreading length $L$, whitening, message bias, and depolarizing noise.

\subsection{Simulation Setup}
\label{Sec:setup}

All simulations were implemented on the full $N$-dimensional Hilbert space. Unless otherwise stated, the experiments use $n=4$ qubits, so that $N=16$. The QHT operator, generator, density matrices, Helstrom measurements, fidelities, and trace distances were computed directly from their matrix definitions.

For each cover state $\ket{c_t}$, Alice applies the signed QHT-domain embedding unitary

\begin{equation}
V_{\alpha_t}(\epsilon)
=
e^{\alpha_t\epsilon G_{\mathcal H}}
=
e^{-\mathrm{i}\alpha_t\epsilon K_{\mathcal H}},
\end{equation}

where $\alpha_t\in\{+1,-1\}$ is determined by the hidden bit after optional whitening. The default cover ensemble consists of Haar-random pure states. For the cover-structure study, we also use $k$-sparse computational-basis cover states.

Bob is assumed to know the expected cover state and applies the optimal binary Helstrom measurement to distinguish the two candidate stego states,
$V_{+}(\epsilon)\ket{c_t}$ and $V_{-}(\epsilon)\ket{c_t}$. When block spreading is used, each encoded bit is embedded over $L$ independently drawn carrier states, and Bob uses majority aggregation over the $L$ per-carrier decisions.

Reliability is measured using Bob's bit error rate. Cover disturbance is measured using the fidelity loss,
$1-F(\ket{c},\ket{s})$. Warden detectability is measured using the trace distance between the stego ensemble and the corresponding cover ensemble.

For noisy-channel experiments, a global depolarizing channel

\begin{equation}
\mathcal{N}_p(\rho)
=
(1-p)\rho+\frac{p}{N}I
\end{equation}

is applied to both cover and stego states before computing Bob's decoding performance and the warden's trace-distance distinguishability.

The Haar-random cover ensemble is used as a generic mathematical stress test for the QHT embedding operator rather than as a literal model of microgrid traffic. Its purpose is to evaluate the average behavior of the signed QHT perturbation over a broad set of quantum states and to verify the expected reliability, disturbance, and trace-distance trends.

\subsection{Numerical Results}
\label{Sec:Numerical}

In this section, we present numerical results to verify the theoretical foundation of the QHT-based steganography protocol empirically.

\begin{figure}[!htbp]
    \centering
    \safeincludegraphics[width=0.75\linewidth]
    {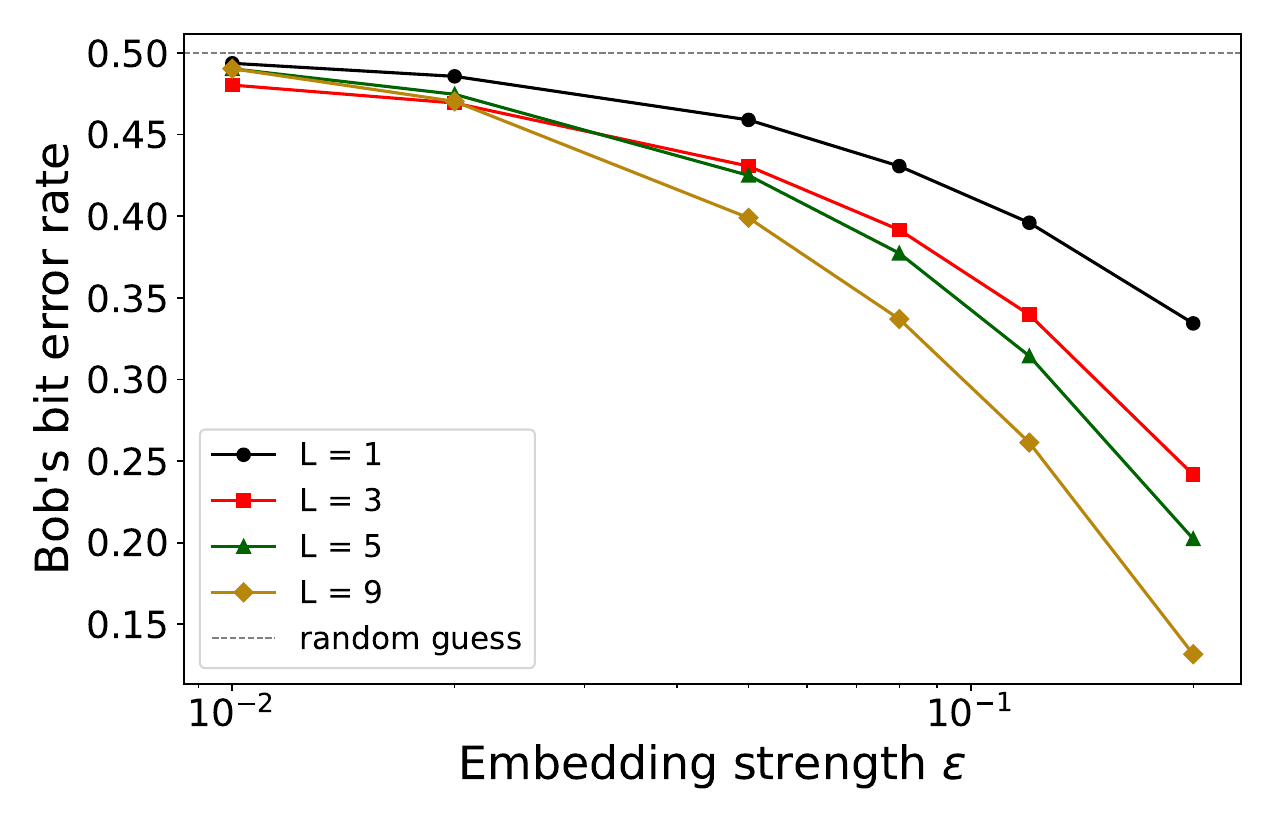}
    \caption{Bob's bit error rate with respect to the embedding strength $\epsilon$ for different spreading lengths $L=1,3,5,9$. For small $\epsilon$, the two signed QHT-domain perturbations are difficult to distinguish, and the error rate remains close to random guessing. As $\epsilon$ increases, the perturbations become more distinguishable and Bob's error rate decreases. Larger spreading lengths further improve reliability through block-majority aggregation.}
    \label{fig:bobber}
\end{figure}

Figure~\ref{fig:bobber} shows the scaling of Bob's detection capability with respect to embedding strength. As expected, BER decreases monotonically with $\epsilon$, and at any fixed $\epsilon$, larger $L$ yields a lower BER. Each message bit is carried by $L$ independent carriers, each on a freshly drawn cover state, and Bob's per-bit decision is the majority vote of the $L$ Helstrom outcomes. For small $\epsilon$, this concentrates the BER approximately as

\begin{equation}
\exp\left(
-2Lc^2\epsilon^2
\right).
\end{equation}

\begin{figure}[!htbp]
    \centering

    \begin{subfigure}[t]{0.49\linewidth}
        \centering
        \safeincludegraphics[width=\linewidth]
        {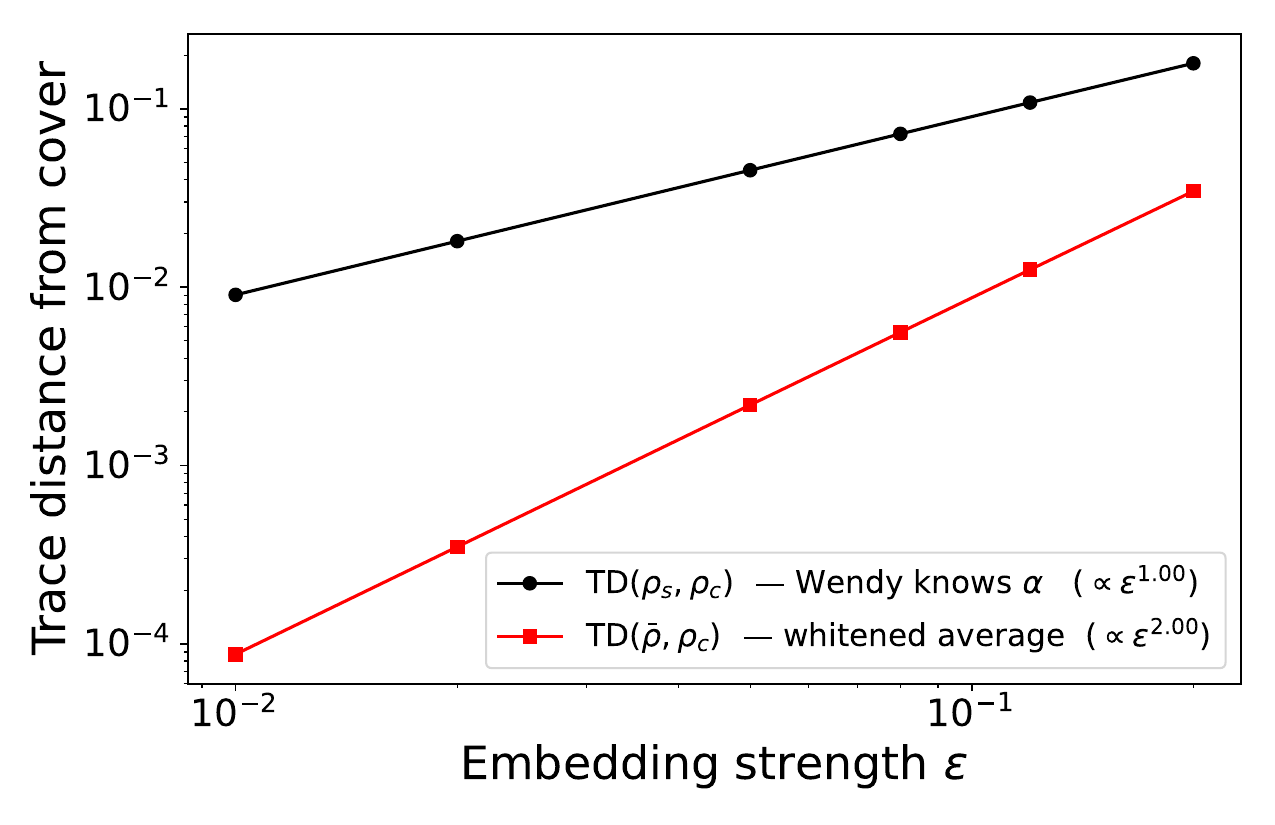}
        \caption{Trace-distance leakage between cover and stego ensembles as a function of embedding strength $\epsilon$. When Wendy knows the embedding sign $\alpha$, the deviation from the cover state scales approximately as $\mathcal{O}(\epsilon)$. Under whitening, the first-order leakage cancels in the averaged stego ensemble, reducing the dominant trace-distance signature to approximately $\mathcal{O}(\epsilon^2)$.}
        \label{fig:covertness}
    \end{subfigure}
    \hfill
    \begin{subfigure}[t]{0.49\linewidth}
        \centering
        \safeincludegraphics[width=\linewidth]
        {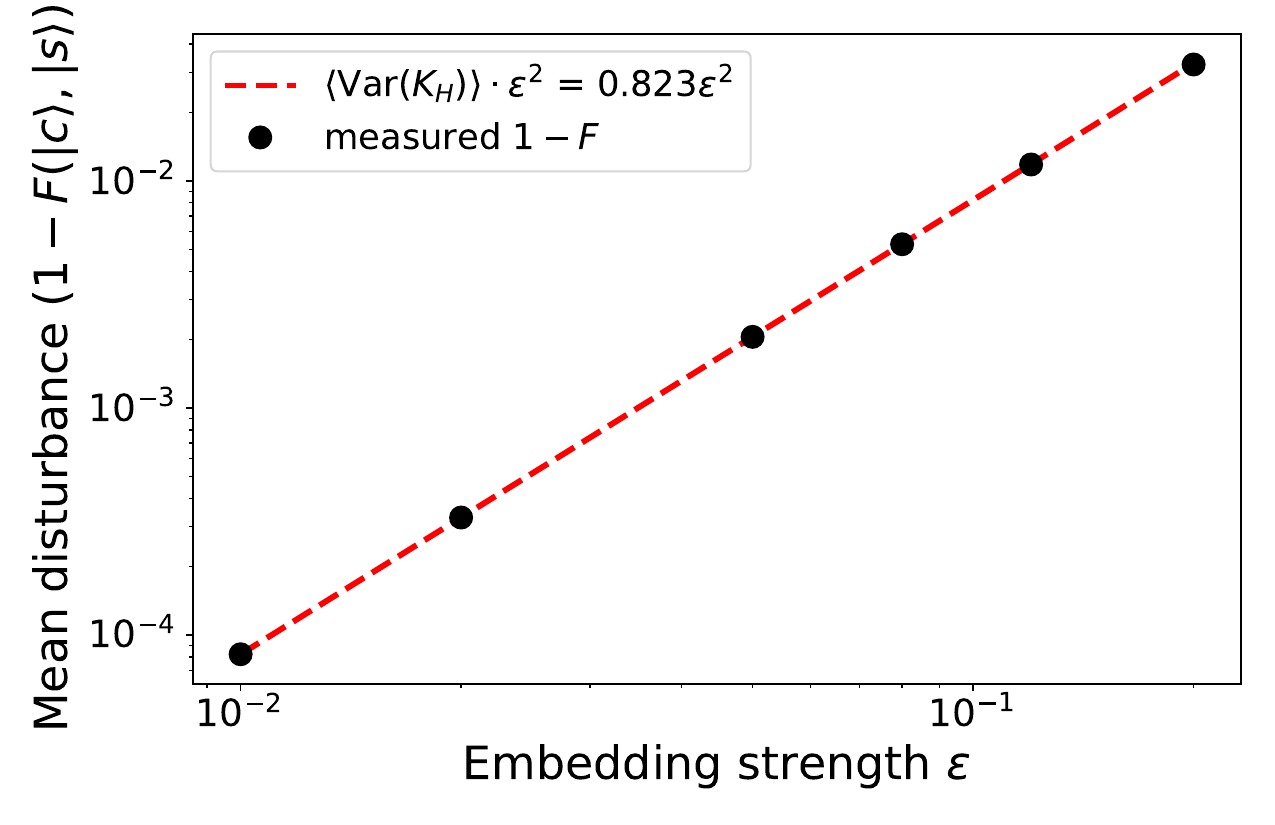}
        \caption{Fidelity disturbance induced by QHT-domain embedding. The measured mean disturbance $1-F(\ket{c},\ket{s})$ closely follows the predicted perturbative scaling $\operatorname{Var}_c(K_{\mathcal H})\epsilon^2$, confirming that the dominant cover-state disturbance is second order in $\epsilon$.}
        \label{fig:fidelityscaling}
    \end{subfigure}

    \caption{Detectability and disturbance footprints of QHT-domain embedding. Panel (a) shows the trace-distance signature observable by Wendy, with and without bit whitening. Panel (b) shows the fidelity disturbance that Alice introduces into the cover state during embedding. Together, the panels show that whitening can eliminate the embedding's first-order $\alpha$-linear footprint, leaving an $\mathcal{O}(\epsilon^2)$ trace-distance residue in the detectability and disturbance budgets.}
    \label{fig:covertness_disturbance}
\end{figure}

Figure~\ref{fig:covertness} plots the trace-distance footprint of a single embedded bit against the embedding strength $\epsilon$ on log--log axes for two warden models: a side-informed warden who knows Alice's bit, represented by the upper curve with slope $1$, and the operational case in which the bit has been XOR-whitened by a fresh pseudo-random key, represented by the lower curve with slope $2$.

The fitted slopes confirm that whitening cancels the first-order $\alpha$-linear leakage of $V_\alpha(\epsilon)$ and leaves only an $\mathcal{O}(\epsilon^2)$ single-carrier trace-distance residue. This result should be interpreted as a per-carrier detectability scaling rather than a full-message covertness guarantee. Since Bob's reliable decoding may also require spreading each encoded bit over $L$ carrier states, the warden's operational detection advantage must account for the accumulated evidence across those carriers. This limitation is important because increasing the spreading length improves Bob's reliability but also increases the number of carrier states available to Wendy for detection.

Figure~\ref{fig:fidelityscaling} represents the mean physical disturbance

\begin{equation}
\left\langle
1-F\left(
\ket{c},
V_\alpha(\epsilon)\ket{c}
\right)
\right\rangle
\end{equation}

that Alice introduces into a random cover state when embedding a single bit at strength $\epsilon$. The measured curve tracks the theoretical prediction with slope $2$ and a matching prefactor across three decades of $\epsilon$. This confirms both that the Fourier-diagonal implementation of $V_\alpha(\epsilon)$ is correct and that $\operatorname{Var}_c(K_{\mathcal H})$ is the scalar controlling the per-carrier disturbance budget. The $\mathcal{O}(\epsilon^4)$ corrections remain negligible throughout the operating range $\epsilon\in[10^{-2},0.2]$ used elsewhere in the manuscript.

\begin{figure}[!htbp]
    \centering
    \safeincludegraphics[width=\linewidth]
    {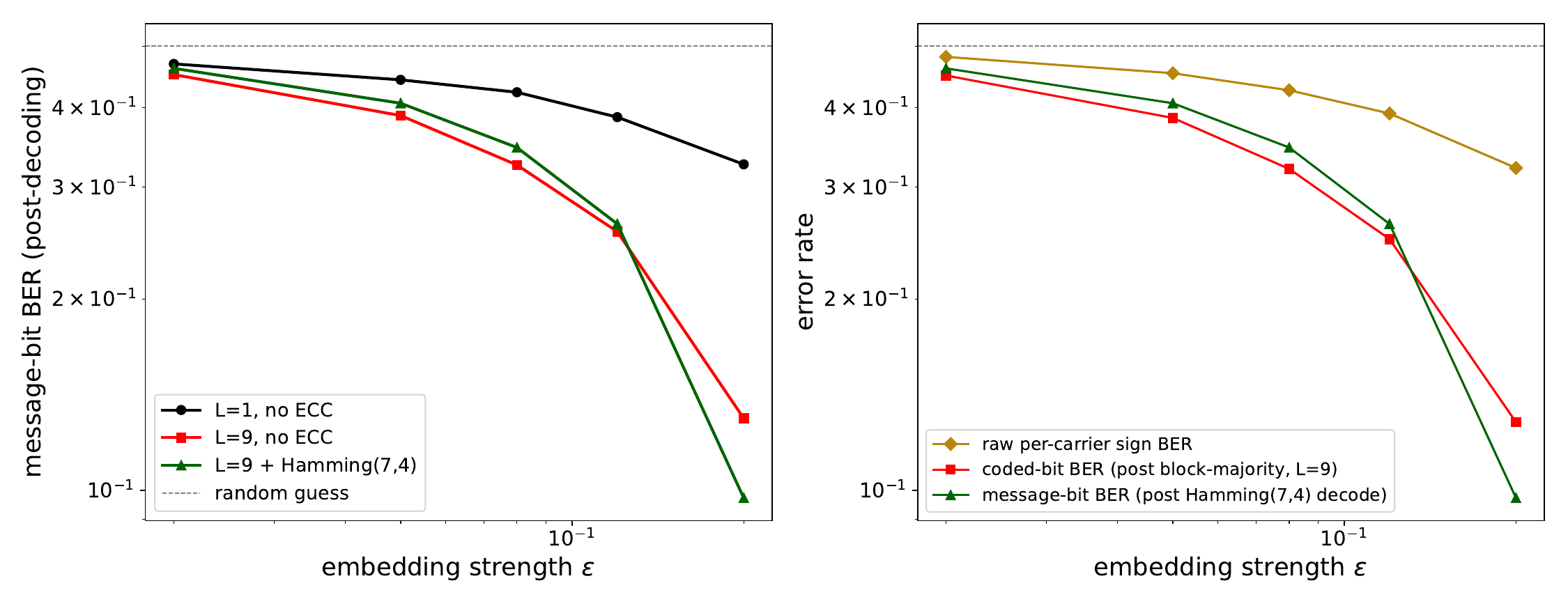}
    \caption{Post-decoding message-bit error results. The left panel shows the post-decoding message-bit error rate as a function of embedding strength $\epsilon$. The right panel shows the stage-by-stage decomposition of error rates for the $L=9$ configuration. The results demonstrate a coding threshold: although classical forward error correction significantly improves reliability when combined with quantum block spreading at robust signal levels, it introduces miscorrection penalties when the raw quantum noise floor exceeds the code's capacity.}
    \label{fig:errorcorrection}
\end{figure}

Figure~\ref{fig:errorcorrection} evaluates the classical error-correcting layer using a Hamming$(7,4)$ code. The left panel compares the post-decoding message-bit error rate for non-encoded transmission with $L=1$, block spreading alone with $L=9$, and the full hybrid scheme with both block spreading and Hamming coding. The right panel shows how the error changes across the $L=9$ pipeline, from raw Helstrom measurements to block-majority decoding and finally to the recovered message bits.

The results show that FEC has a clear threshold in high-noise quantum regimes. At sufficiently large embedding strength $\epsilon$, the Hamming code works effectively with block spreading and significantly reduces the final error rate. However, when $\epsilon$ is too small, the raw quantum error is too high for the classical code to correct reliably. Future work can investigate whether more modern error-correcting codes can mitigate this limitation.

\begin{figure}[!htbp]
    \centering
    \safeincludegraphics[width=\linewidth]
    {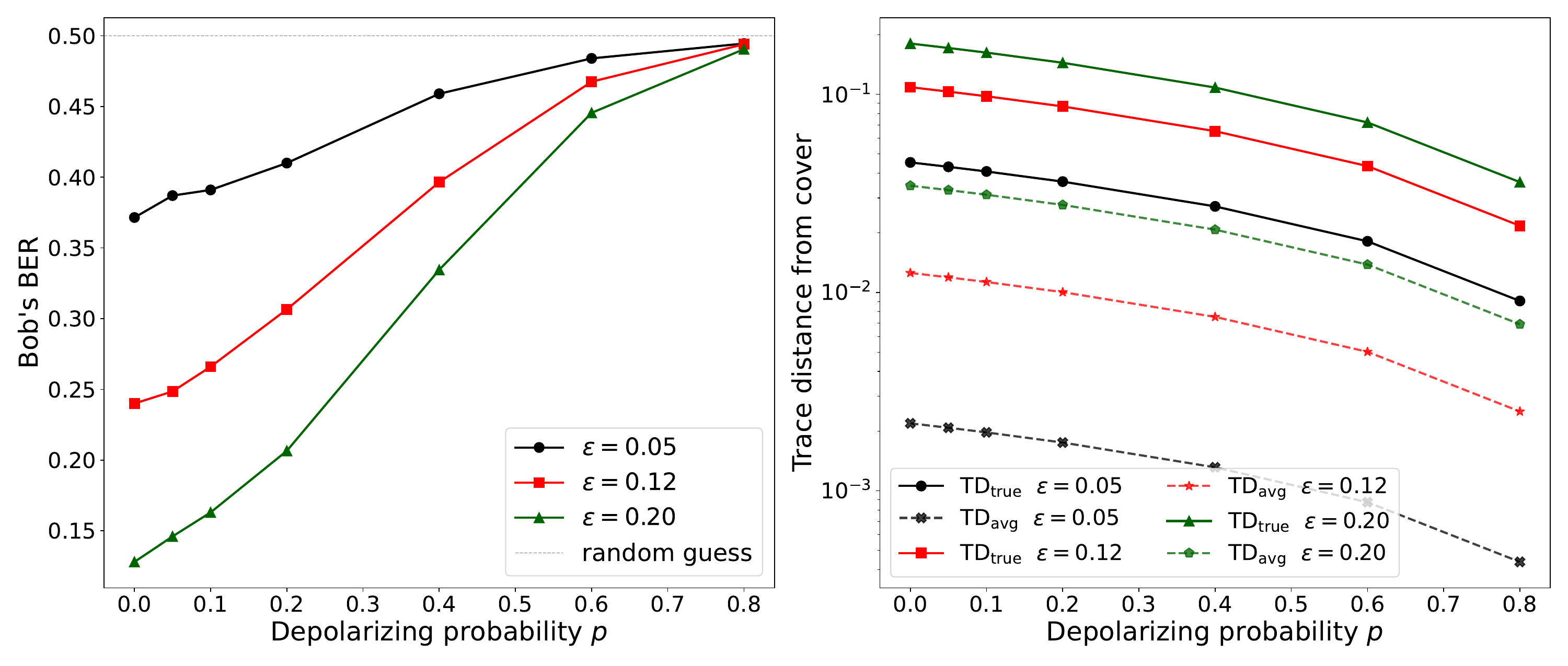}
    \caption{Effect of depolarizing noise on the reliability and detectability of the protocol. The left panel shows Bob's bit error rate for different embedding strengths $\epsilon$, while the right panel shows trace-distance leakage for known-sign and whitened-average stego ensembles. Increasing depolarizing noise degrades Bob's decoding reliability but also reduces Wendy's ability to distinguish stego states from cover states. This illustrates the tradeoff among embedding strength, recoverability, channel noise, and covertness.}
    \label{fig:noise}
\end{figure}

Figure~\ref{fig:noise} shows the effect of depolarizing noise on Bob's decoding reliability and on stego-cover distinguishability. In the left panel, Bob's BER increases as the depolarizing probability $p$ increases, approaching the random-guess limit of $0.5$. This behavior is expected because depolarizing noise progressively removes the phase-domain encoded information used to distinguish the positive and negative perturbations.

Larger embedding strengths remain more robust. For example, $\epsilon=0.20$ maintains a substantially lower BER than $\epsilon=0.05$ for small and moderate noise levels, but all curves eventually move toward random guessing as the channel becomes highly noisy.

The right panel shows the corresponding trace-distance behavior. Both $TD_{\mathrm{true}}$, which measures the distinguishability of a signed stego state from the cover state, and $TD_{\mathrm{avg}}$, which measures the distinguishability of the whitened average state from the cover state, decrease as $p$ increases. This reflects the contraction of distinguishability under depolarizing noise. However, this reduction in detectability is not a free security gain because the same noise also degrades Bob's ability to recover the hidden message. The gap between $TD_{\mathrm{true}}$ and $TD_{\mathrm{avg}}$ remains evident across the entire noise range, indicating that whitening continues to suppress first-order signed leakage under depolarizing noise.

\section{Discussion and Conclusions}
\label{Conclusions}

In this work, we introduced a finite-dimensional quantum Hilbert transform (QHT) as a unitary analogue of the classical Hilbert transform. The proposed construction uses the quantum Fourier transform to map a quantum state into the Fourier basis, applies a frequency-sign-dependent phase rule, and then maps the state back to the computational basis. Unlike the classical Hilbert transform, which may suppress the zero-frequency component, the proposed QHT preserves the DC and Nyquist modes so that the overall operation remains unitary. This allows the QHT to retain the phase-shifting structure of the classical Hilbert transform while remaining compatible with coherent quantum evolution.

Building on this operation, we proposed a quantum steganography protocol in which a classical hidden message is embedded into weakly signed phase perturbations of legitimate quantum cover states. The hidden bit is represented by the sign of a small QHT-domain unitary displacement. Bob recovers the embedded information by distinguishing between positive and negative perturbations using a binary quantum state-discrimination procedure, followed by block-level aggregation and classical error-correcting decoding.

Beyond steganography, the protocol can be used to develop possible attack vectors for quantum-augmented networks (QuANets). QuANets integrate quantum communication primitives with classical network infrastructure to develop more secure network architectures \citep{jha2024ml,jha2025towards,jha2025quantum}. Within such a setting, the steganography protocol could be adapted into an attack vector in which compromised nodes introduce weak QHT-domain perturbations to leak low-rate information while maintaining the appearance of normal protocol execution. The QHT-based embedding model therefore provides a controlled way to define and evaluate phase-domain covert-channel attacks.

This is particularly relevant for quantum-augmented cyber-physical systems such as microgrids \citep{jha2026novel}. In a quantum-augmented microgrid, quantum primitives may be used to strengthen communication among controllers, sensors, distributed energy resources, and supervisory nodes. However, if a trusted or semi-trusted node is compromised, the same quantum communication layer may create a new attack surface. The attacker may not need to break the quantum protocol directly. Instead, the attacker may exploit the legitimate quantum carrier by adding small structured perturbations that encode side information, operational state, timing information, or control-related metadata. Such a channel could remain below ordinary error-rate or fidelity thresholds while still being decodable by the intended covert receiver.

The QHT framework therefore provides a method for constructing and studying such attack vectors. The embedding strength $\epsilon$, spreading length $L$, whitening sequence, and cover-state structure can be treated as adversarial design parameters. By varying these parameters, one can study how much information can be hidden before the perturbation becomes detectable by a network monitor. This makes the proposed protocol useful both for demonstrating quantum steganography and for systematically evaluating covert leakage risks in quantum-augmented network architectures. One focus of future work will be to modify this protocol to construct a formal covert-channel attack model for QuANets.

Another future direction is to study the effects of channel-noise models other than depolarization on the QHT implementation and the quantum-steganography protocol. This would support the development of a hardware-level implementation framework that can be adapted to practical use cases. These directions would allow the proposed QHT-based protocol to serve as both a quantum steganographic primitive and a structured attack model for future quantum-augmented infrastructure.

\bibliographystyle{apalike}
\bibliography{main}

\begin{thebibliography}{}

\bibitem[Deng, 2006]{hilbertHiding2006}
Deng, L. (2006).
\newblock A new approach of data hiding within speech based on hash and {Hilbert} transform.
\newblock In {\em 2006 International Conference on Systems and Networks Communications ({ICSNC}'06)}, pages 13--13.

\bibitem[Feldman, 2008]{Feldman2008TheoreticalAA}
Feldman, M. (2008).
\newblock Theoretical analysis and comparison of the {Hilbert} transform decomposition methods.
\newblock {\em Mechanical Systems and Signal Processing}, 22:509--519.

\bibitem[Hahn, 1996]{Hahn1996HilbertTI}
Hahn, S.~L. (1996).
\newblock {\em {Hilbert} transforms in signal processing}.
\newblock Artech House, Norwood, MA.

\bibitem[Hameed et~al., 2022]{hameed2022literature}
Hameed, R.~S., Abd~Rahim, B. H.~A., Taher, M.~M., and Mokri, S.~S. (2022).
\newblock A literature review of various steganography methods.
\newblock {\em Journal of Theoretical and Applied Information Technology}, 100(5):1412--1427.

\bibitem[Huang, 2014]{hilberthuang}
Huang, N.~E. (2014).
\newblock {\em {Hilbert--Huang} transform and its applications}, volume~16.
\newblock World Scientific, Singapore.

\bibitem[Jha and Parakh, 2025]{jha2025quantumhilbert}
Jha, N. and Parakh, A. (2025).
\newblock {Quantum Hilbert Transform}.

\bibitem[Jha et~al., 2024]{jha2024ml}
Jha, N., Parakh, A., and Subramaniam, M. (2024).
\newblock A {ML}-based approach to quantum-augmented {HTTP} protocol.
\newblock In {\em 2024 {IEEE} International Conference on Quantum Computing and Engineering ({QCE})}, volume~2, pages 591--592. IEEE.

\bibitem[Jha et~al., 2025a]{jha2025quantum}
Jha, N., Parakh, A., and Subramaniam, M. (2025a).
\newblock Quantum key distribution: Bridging theoretical security proofs, practical attacks, and error correction for quantum-augmented networks.
\newblock {\em Cryptologia}, pages 1--58.

\bibitem[Jha et~al., 2025b]{jha2025towards}
Jha, N., Parakh, A., and Subramaniam, M. (2025b).
\newblock Towards a quantum-classical augmented network.
\newblock In {\em Quantum Computing, Communication, and Simulation V}, volume 13391, pages 72--86. SPIE.

\bibitem[Jha et~al., 2026]{jha2026novel}
Jha, N., Paudel, P., Parakh, A., and Subramaniam, M. (2026).
\newblock A novel quantum augmented framework to improve microgrid cybersecurity.

\bibitem[Kak, 1977]{kak1977discrete}
Kak, S. (1977).
\newblock The discrete finite {Hilbert} transform.
\newblock {\em Indian Journal of Pure and Applied Mathematics}, 8:1385--1390.

\bibitem[Kak, 1970]{Kak1970DHT}
Kak, S.~C. (1970).
\newblock The discrete {Hilbert} transform.
\newblock {\em Proceedings of the IEEE}, 58(4):585--586.

\bibitem[Kandregula, 2009]{kandregula2009}
Kandregula, R. (2009).
\newblock The basic discrete {Hilbert} transform with an information hiding application.

\bibitem[Khare et~al., 2022]{khare2022vhers}
Khare, S.~K., Gaikwad, N.~B., and Bajaj, V. (2022).
\newblock {VHERS}: A novel variational mode decomposition and {Hilbert} transform-based {EEG} rhythm separation for automatic {ADHD} detection.
\newblock {\em IEEE Transactions on Instrumentation and Measurement}, 71:1--10.

\bibitem[Kschischang, 2006]{Hilbert2006}
Kschischang, F.~R. (2006).
\newblock The {Hilbert} transform.
\newblock Technical Report~83, University of Toronto.

\bibitem[Li et~al., 2011]{Li2011}
Li, Z., Chi, H., Zhang, X., and Yao, J. (2011).
\newblock Optical single-sideband modulation using a fiber-{Bragg}-grating-based optical {Hilbert} transformer.
\newblock {\em IEEE Photonics Technology Letters}, 23(9):558--560.

\bibitem[Motlagh and Wiebe, 2024]{motlagh2024generalized}
Motlagh, D. and Wiebe, N. (2024).
\newblock Generalized quantum signal processing.
\newblock {\em PRX Quantum}, 5(2):020368.

\bibitem[Oppenheim and Schafer, 2021]{oppenheim2021discrete}
Oppenheim, A.~V. and Schafer, R.~W. (2021).
\newblock {\em Discrete-time signal processing}.
\newblock Pearson, Upper Saddle River, NJ, 4 edition.

\bibitem[Padala and Prabhu, 1997]{DHT2}
Padala, S.~K. and Prabhu, K. M.~M. (1997).
\newblock Systolic arrays for the discrete {Hilbert} transform.
\newblock {\em IEE Proceedings - Circuits, Devices and Systems}, 144(5):259--264.

\bibitem[Pang et~al., 2019]{pang2019signal}
Pang, C.-Y., Zhou, R.-G., Hu, B.-Q., Hu, W.-W., and El-Rafei, A. (2019).
\newblock Signal and image compression using quantum discrete cosine transform.
\newblock {\em Information Sciences}, 473:121--141.

\bibitem[Sharma et~al., 2022]{sharma2022hilbert}
Sharma, V.~K., Sharma, P.~C., Goud, H., and Singh, A. (2022).
\newblock {Hilbert} quantum image scrambling and graph signal processing-based image steganography.
\newblock {\em Multimedia Tools and Applications}, 81(13):17817--17830.

\bibitem[Shor, 1999]{shor1999polynomial}
Shor, P.~W. (1999).
\newblock Polynomial-time algorithms for prime factorization and discrete logarithms on a quantum computer.
\newblock {\em SIAM Review}, 41(2):303--332.

\bibitem[Tan et~al., 2007]{tan2007steganalysis}
Tan, S., Huang, J., and Shi, Y.~Q. (2007).
\newblock Steganalysis of enhanced {BPCS} steganography using the {Hilbert--Huang} transform based sequential analysis.
\newblock In {\em International Workshop on Digital Watermarking}, pages 112--126, Berlin, Heidelberg. Springer.

\bibitem[Wu et~al., 2011]{wu2011detection}
Wu, S., Li, W., and Shi, Y. (2011).
\newblock Detection for steganography based on {Hilbert--Huang} transform.
\newblock In {\em 2011 International Conference on Photonics, 3D-Imaging, and Visualization}, volume 8205, pages 262--267. SPIE.

\bibitem[Yin et~al., 2021]{yin2021quantum}
Yin, H., Lu, D., and Zhang, R. (2021).
\newblock Quantum windowed {Fourier} transform and its application to quantum signal processing.
\newblock {\em International Journal of Theoretical Physics}, 60:3896--3918.

\bibitem[Zhou et~al., 2017]{zhouquant2017}
Zhou, S., Loke, T., Izaac, J.~A., and Wang, J.~B. (2017).
\newblock Quantum {Fourier} transform in computational basis.
\newblock {\em Quantum Information Processing}, 16(3):82.

\end{thebibliography}

\section*{Author Contributions}

N.J. is the primary author of the manuscript and received intellectual inputs from A.P.

\section*{Data Availability Statement}

All data generated or analyzed during this study are included in this article and its supplementary information files.

\section*{Competing Interests}

The authors declare no competing interests.

\section*{Funding}

The authors received no specific funding for this work.

\appendix

\section{Proof of Unitarity of the QHT}
\label{appendixA}

Since we propose the QHT as a quantum operator, we must show that it is unitary.

\begin{proposition}
The quantum Hilbert transform $U_{\mathcal H}$ is unitary.
\end{proposition}

\begin{proof}
The QFT operator $F_N$ is unitary, so

\begin{equation}
F_N^\dagger F_N=I.
\end{equation}

The operator $D_{\mathcal H}$ is diagonal, and every diagonal entry has unit magnitude:

\begin{equation}
|h_k|=1 \qquad \forall\, k. 
\end{equation}

Therefore,

\begin{equation}
D_{\mathcal H}^\dagger D_{\mathcal H}
=
I.
\end{equation}

Now,

\begin{align}
U_{\mathcal H}^\dagger U_{\mathcal H}
&=
\left(
F_N^\dagger D_{\mathcal H}F_N
\right)^\dagger
\left(
F_N^\dagger D_{\mathcal H}F_N
\right)\\
&=
F_N^\dagger
D_{\mathcal H}^\dagger
F_NF_N^\dagger
D_{\mathcal H}
F_N\\
&=
F_N^\dagger
D_{\mathcal H}^\dagger
D_{\mathcal H}
F_N\\
&=
F_N^\dagger F_N\\
&=
I.
\end{align}

Hence, $U_{\mathcal H}$ is unitary.
\end{proof}

\section{Helstrom Measurement for Decoding}
\label{app:helstrom-decoding}

For the $t$-th carrier state, Bob must distinguish between the two candidate stego states

\begin{equation}
\ket{\psi_{+}^{(t)}}
=
V_{+}(\epsilon)\ket{c_t},
\qquad
\ket{\psi_{-}^{(t)}}
=
V_{-}(\epsilon)\ket{c_t}.
\end{equation}

The corresponding density operators are

\begin{equation}
\rho_{+}^{(t)}
=
\ket{\psi_{+}^{(t)}}
\bra{\psi_{+}^{(t)}},
\qquad
\rho_{-}^{(t)}
=
\ket{\psi_{-}^{(t)}}
\bra{\psi_{-}^{(t)}}.
\end{equation}

Assuming equal prior probabilities, the optimal minimum-error binary measurement is the Helstrom measurement. Define

\begin{equation}
\Delta_t
=
\frac{1}{2}\rho_{+}^{(t)}
-
\frac{1}{2}\rho_{-}^{(t)}.
\end{equation}

The optimal POVM is

\begin{equation}
\left\{
\Pi_{+}^{(t)},
\Pi_{-}^{(t)}
\right\},
\end{equation}

where $\Pi_{+}^{(t)}$ is the projector onto the positive eigenspace of $\Delta_t$, and

\begin{equation}
\Pi_{-}^{(t)}
=
I-\Pi_{+}^{(t)}.
\end{equation}

Bob decides in favor of the positive embedding sign when the outcome associated with $\Pi_{+}^{(t)}$ occurs and decides in favor of the negative embedding sign otherwise. The corresponding minimum probability of error is

\begin{equation}
P_{e,t}^{\star}
=
\frac{1}{2}
\left(
1-
\sqrt{
1-
\left|
\left\langle
\psi_{+}^{(t)}
\middle|
\psi_{-}^{(t)}
\right\rangle
\right|^2
}
\right).
\end{equation}

For the QHT-domain embedding operation,

\begin{equation}
V_{+}(\epsilon)
=
\exp\left(
\epsilon G_{\mathcal H}
\right),
\qquad
V_{-}(\epsilon)
=
\exp\left(
-\epsilon G_{\mathcal H}
\right).
\end{equation}

Since $G_{\mathcal H}$ is skew-Hermitian,

\begin{equation}
V_{+}^{\dagger}(\epsilon)
=
V_{-}(\epsilon),
\end{equation}

and therefore,

\begin{equation}
\left\langle
\psi_{+}^{(t)}
\middle|
\psi_{-}^{(t)}
\right\rangle
=
\bra{c_t}
\exp\left(
-2\epsilon G_{\mathcal H}
\right)
\ket{c_t}.
\end{equation}

Thus, the ideal decoding error is controlled by the overlap between the two QHT-perturbed states. As $\epsilon$ increases, the two states become more distinguishishable, reducing the ideal decoding error.

\end{document}